\documentclass[11pt]{article}
\usepackage{xcolor}
\usepackage[colorlinks=true,citecolor=black,urlcolor=gray, filecolor=black, linkcolor=black]{hyperref}
\usepackage{fullpage, color}
\usepackage{amsmath,amssymb,amsfonts}
\usepackage{setspace}
\usepackage{bbm}

\allowdisplaybreaks

\newcommand{\al}{\alpha}
\newcommand{\eps}{\epsilon}

\newcommand{\Z}{\mathbb{Z}}
\newcommand{\R}{\mathbb{R}}
\newcommand{\F}{\mathbb{F}}

\newcommand{\calp}{\mathcal{P}}

\newcommand{\calf}{\mathcal{F}}
\newcommand{\calg}{\mathcal{G}}
\newcommand{\call}{\mathcal{L}}

\newcommand{\E}{\operatornamewithlimits{\mathbb{E}}}
\newcommand{\1}{{\bf 1}}
\newcommand{\supp}{\mathrm{supp}}
\newcommand{\sspan}{\mathrm{span}}

\newcommand{\ra}{\rightarrow}
\newcommand{\Var}{\mathrm{Var}}

\newcommand{\eqdef}{{\stackrel{\rm def}{=}}}

\newcommand{\poly}{{\rm {poly}}}

\newcommand{\mspan}{{\rm span}}
\newcommand{\cube}{\mathbb{F}_{2}^{n}}

\newcommand{\zo}{\{0,1\}}
\newcommand{\ocf}{OCF}
\newcommand{\ignore}[1]{}

\newtheorem{theorem}{Theorem}
\newtheorem{corollary}[theorem]{Corollary}
\newtheorem{lemma}[theorem]{Lemma}

\newtheorem{definition}[theorem]{Definition}
\newtheorem{claim}[theorem]{Claim}
\newtheorem{fact}[theorem]{Fact}

\newtheorem{question}[theorem]{Question}

\newcommand{\qed}{\hfill \ensuremath{\Box}}

\newenvironment{proof}{\noindent{\bf Proof}\hspace*{1em}}{\qed\bigskip}
\newenvironment{proof-of-lemma}[1]{\noindent{\bf Proof of Lemma #1}\hspace*{1em}}{\qed\bigskip}
\newenvironment{proof-of-theorem}[1]{\noindent{\bf Proof of Theorem #1}\hspace*{1em}}{\qed\bigskip}

\newcommand{\email}[1]{\href{mailto:#1}{{\tt #1}}}

\def\mnotes{0} 

\ifnum\mnotes=0
\newcommand{\mnote}[1]{}
\else
	\newcounter{mynotes}
	\newcommand{\mnote}[1]{\addtocounter{mynotes}{1}{{\bf !}}
	\marginpar{\scriptsize  {\arabic{mynotes}.\ {\sf \textcolor{blue}{#1}}}}}
	\fi

\newcommand{\enote}[1]{\mnote{EG:\\ #1}}
\newcommand{\abnote}[1]{\mnote{AB: #1}}
\title{Testing Odd-Cycle-Freeness in Boolean Functions}
\author{Arnab Bhattacharyya\thanks{CSAIL, MIT, Cambridge, MA,
 \email{abhatt@mit.edu}. Supported in part by NSF grants 1065125 and 0728645.}
\and Elena Grigorescu\thanks{ Georgia Tech, Atlanta, GA, \email{elena@cc.gatech.edu}.
Supported by  NSF award 1019343 to the Computing Research Association for the Computing Innovation Fellowship Program.}
\and Prasad Raghavendra\thanks{Georgia Tech, Atlanta, GA, \email{raghavendra@cc.gatech.edu}.}
\and Asaf Shapira\thanks{Georgia Tech, Atlanta, GA, \email{asafico@math.gatech.edu}. Supported in part by NSF Grant DMS-0901355.}}

\date{}

\begin{document}
\maketitle
\begin{abstract}

Call a function $f: \F_2^n \to \{0,1\}$ {\em odd-cycle-free} if  there are no $x_1,
\dots, x_k \in \F_2^n$ with $k$ an odd integer such that $f(x_1) = \cdots =
f(x_k) = 1$ and $x_1 + \cdots + x_k = 0$. We show that one can distinguish
odd-cycle-free functions from those $\eps$-far from being odd-cycle-free by making
$\poly(1/\eps)$ queries to an evaluation oracle. To obtain this result, we use connections 
between basic Fourier analysis and spectral graph theory to show that one can reduce testing
odd-cycle-freeness of Boolean functions to testing bipartiteness of dense graphs.
Our work forms part of a recent sequence of works that shows connections between
testability of properties of Boolean functions and of graph properties.

\ignore{
Our work forms part of a recent sequence of works that shows connections between
testability of properties of Boolean functions and of graph properties. The graph property
analogous to odd-cycle-freeness is bipartiteness. Even though for general monotone graph
properties, one can bound the query complexity only by a very fast-growing function of $1/\eps$,
it is known that bipartiteness can be tested using $\tilde{O}(1/\eps^2)$ queries
(Alon et al., SIDMA '02). Our work shows that a similar phenomenon takes place for
linear-invariant properties. Odd-cycle-freeness was known to be testable previously
(Bhattacharyya et al., FOCS '10) but the only bound on the query complexity was a tower of
exponentials in $1/\eps$. Here, we show that it is possible to test odd-cycle-freeness
using $\tilde{O}(1/\eps^2)$ queries.  We prove this by analyzing the eigenvalues of the
Cayley graph associated with a Boolean function.
}

We also prove that there is a {\em canonical} tester for odd-cycle-freeness making
$\poly(1/\eps)$ queries, meaning that the testing algorithm operates by picking a random
linear subspace of dimension $O(\log 1/\eps)$ and then checking if the restriction of the
function to the subspace is odd-cycle-free or not.  The test is analyzed by studying the
effect of random subspace restriction on the Fourier coefficients of a function. 
Our work implies that testing odd-cycle-freeness using a canonical tester instead of an arbitrary tester incurs no more than a polynomial blowup in the query complexity.
The question of whether a canonical tester with polynomial blowup exists for all
linear-invariant properties remains an open problem.

\ignore{

Recent developments have shown that tools and approaches that have been very successful in
designing algorithms for testing graph properties can also be used in testing properties
of Boolean functions. We continue this line of research and obtain the following results:

\begin{itemize}
\item The property of being Odd-Cycle-Free (\ocf) has been extensively studied in many models of property testing.
In the context of Boolean functions, we say that a function is \ocf\ if its support does not contain an odd set of points
$x_1,\ldots,x_{2k+1}$ satisfying $\sum_ix_i=0$. Our first result is an efficient testing algorithm for \ocf\ of Boolean functions.
This algorithm uses properties of the Fourier transform of a boolean function together with an analysis of the eigenvalues of the Cayley graph associated with a Boolean function.

\item
The above algorithm is not ``canonical'' in the sense that it does not work by sampling a set of points $S$ and asking about all the points in their span.
Motivated by the fact that testers for graph properties can be made canonical with only a small blow-up in the query complexity, we investigate the problem of turning the previous \ocf\ algorithm into a canonical one. Using more involved methods, we show that \ocf\ can be tested by sampling $O(\log(1/\epsilon))$ points and querying $f$ on their span. The general problem of proving the existence of a canonical tester for testing properties of Boolean functions with polynomial blowup remains open.
\end{itemize}}
\end{abstract}


\newcommand{\keywords}[1]{\noindent {\bf Keywords:} #1}


\section{Introduction}\label{sec:intro}

A property testing algorithm is required to distinguish objects that
satisfy a given property from objects that are ``far'' from satisfying the property.
One can trace the origins of property testing as an area of study to two distinct origins:
\cite{BLR} (and subsequently \cite{RS}) which formally investigated the testability of
linearity and other {\em properties of Boolean functions} and \cite{GGR} which
studied the testability of various {\em graph properties}. Although \cite{GGR} was
inspired by the preceding work on Boolean functions, the two directions evolved more or
less independently  in terms of the themes considered and the techniques employed.
Recently though, this has dramatically changed, and quite a few
surprising connections have emerged.
In this work, we draw more connections between these two apparently different areas and
show how ideas and tools used in the study of graph properties can be used to test certain
properties of Boolean functions.

We start with a few definitional remarks. A property of Boolean functions  is a subset
$\calp \subseteq \{ f:\{0,1\}^n\ra\{0,1\} \}$. The distance between  $f, g: \{0,1\}^n\ra
\{0,1\}$ is given by the Hamming metric $\delta(f, g)=\Pr_x[f(x)\not = g(x)],$ and the
distance from $f$ to $\calp$ is $\delta_{\calp}(f)=\min_{g \in \calp} \delta(f, g)$.  A
function $f$ is $\eps$-far from $\calp$ if $\delta_{\calp}(f)\geq \eps$.
These definitions carry over to graph properties\footnote{In this paper, when we refer to
  graph properties, we will always mean the dense graph model where the graph is
  represented by its adjacency matrix.}, where the distance to a graph property
$\calp$ is said to be $\eps$ if $\eps n^2$ edges need to be added to or removed from the
given graph on $n$ vertices in order to obtain a graph in $\calp$. A {\em tester} for
$\calp$ is a randomized algorithm which, given oracle  access to the input and a parameter
$\eps \in (0,1)$, accepts with probability at least $2/3$ when the input is in $\calp$ and
rejects with probability at least $2/3$ when it is $\eps$-far from $\calp$.
In the case of Boolean functions, the tester can query the value of the function at any
element of $\zo^n$, and in the case of graphs, it can query the adjacency matrix at any
location. The
complexity of a tester is measured by the number of queries it makes to the
oracle, and if this quantity is  independent of $n$, the property is called {\em (strongly)
testable}. A one-sided error tester should accept every object in $\calp$ with probability
$1$ and reject every object that is $\eps$-far from $\calp$  with probability $2/3$.

Our main focus in this paper is the study of the following property of Boolean functions.
\begin{definition}[Odd-Cycle-Freeness]
A function $f:\F_2^n \to \{0,1\}$ is said to be {\em odd-cycle-free (OCF)} if for all odd
$k\geq 1$, there are no $x_1, x_2, \ldots, x_{k} \in \F_2^n$ satisfying
$x_1+\cdots+x_k=0$ and $f(x_i)=1$ for all $i\in [k]$.
\end{definition}

The name ``odd-cycle-freeness'' arises from the observation that $f$ is \ocf\ if and only
if the Cayley graph\footnote{See Section~\ref{sec:edge-sample} for the precise definition.}
with the support of $f$ as its generators is free of cycles of odd length, i.e. is bipartite.
The property of bipartiteness in general graphs has been extensively studied, and nearly tight upper
and lower bounds are known for its query complexity \cite{GGR, AlonK02,BT04, KaufmanKR04}.
In this work, we show that odd-cycle-freeness for Boolean functions
is testable with comparable query complexity and moreover, using tests that are very
similar to the ones for graph bipartiteness.

Odd-cycle-freeness can also be described in a more algebraic way. As observed in
Section \ref{sec:charac}, given a function $f: \F_2^n \to \zo$ with  density $\rho \eqdef
\E_x[f(x)]$, the distance of $f$ to \ocf\ is exactly $\frac{1}{2}(\rho + \min_\alpha
\hat{f}(\alpha))$.  So, a Boolean function's distance to \ocf\ is directly connected to
the (signed) value of its smallest Fourier coefficient. This link
proves crucial in our analysis of tests for \ocf.

Our work is part of a larger program to understand the structure of testable properties
of Boolean functions. We explain this perspective next.

\paragraph{Common themes in testing}

 A leading question in the search for common unifying themes in property testing has been
 that of discovering necessary and sufficient conditions for strong testability.
Kaufman and Sudan~\cite{KS} suggest that {\em linear invariance} is a natural property of
boolean functions and play an important role in testing.
%
Formally, a property\footnote{Henceforth, we will identify $\{0,1\}^n$ with the vector
  space $\F_2^n$.} $\calp \subseteq \{ f:\F_2^n\ra\{0,1\} \}$ is said to be
linear-invariant if for any $f\in \calp$, it is also the case that $f\circ L \in \calp$,
for any $\F_2$-linear transformation $L: \F_2^n\ra \F_2^n$.
Some notable examples of properties that were shown to be testable and which are invariant under linear transformations of the domain include  linear functions~\cite{BLR},  low degree polynomials~\cite{AKKLR}, and functions with low Fourier dimensionality or sparsity~\cite{GOSSW}.

 A  general class of linear invariant families can be described in terms of forbidden
 patterns. The first instance of this perspective appeared in the work of
 Green~\cite{Green05} in testing if a Boolean property is {\em triangle-free}. Formally, $f$ is
 triangle-free if it is free from the pattern $\langle f(x), f(y), f(x+y) \rangle =\langle
 1^3\rangle,$ for any $x, y\in \F_2^n$.
\ignore{
This can be equivalently formulated as saying that
 for any linear map $L:\F_2^n\ra \F_2^n$ $f$ is free from the pattern $\langle f(L(e_1)),
 f(L(e_2)), f(L(e_1+e_2))\rangle=\langle1^3\rangle$, where $e_1, e_2$ are the first two
 standard basis vectors over $\F_2^n$. Hence, if $\calp$ is a linear invariant family
 characterized by a forbidden pattern,  the forbidden pattern can be encoded by the linear
 dependencies between the elements involved, while the actual identities of these elements
 is not relevant. }
More generally, a function is said to be free from solutions to the linear equation $x_1 +
\cdots + x_k = 0$ if there are no $x=(x_1, \ldots, x_k)\in (\F_2^n)^k$ satisfying $x_1 + \cdots +
x_k=0$ and $f(x_i)=1$ for all $i\in [k]$.
Pushing this generalization further, for a matrix $M\in \F_2^{m\times k}$ and $\sigma\in
\{0,1\}^k$,  we say that $f$ is $(M, \sigma)$-free if there is no $x=(x_1, \ldots,
x_k)\in (\F_2^n)^k$ such that $Mx=0$ and $f(x_i)=\sigma_i$ for all $i\in [n].$ This
corresponds to freeness from solutions to a system of linear equations. When
$\sigma=1^k$, notice that if $f\in \calp$, then any function obtained from $f$ by removing
elements in the support of $f$ also must belong to $\calp$, and in this case $\calp$ is
said to be {\em monotone}.

Green~\cite{Green05} proved that $(M,  1^k)$-freeness is testable with one sided error
when  $M$ is a rank $1$ matrix. Kr\'al{'}, Serra and Vena~\cite{KSV08b} and
Shapira~\cite{Shap09} showed that this is true regardless of the matrix $M$. The setting
when $\sigma\not = 1^k$ was introduced in~\cite{BCSX09}, where it is shown that
$(M,\sigma)$-freeness is testable for any $\sigma$ as long as $M$ is of rank one.
The case of $M$ being rank $1$ was fully resolved by Bhattacharyya, Grigorescu and Shapira
~\cite{BGS10} who
showed that
 any
(possibly infinite) intersection of such properties is also testable. Now, note that
odd-cycle-freeness is an example of such a property; it is the intersection of
$(C_k,1^k)$-freeness for all odd $k \geq 1$, where $C_k$ is the row vector $[1~1~ \cdots~ 1]$
of size $k$.  We next state this result formally. 
\begin{theorem}[\cite{BGS10}]\label{thm:testable}
There exists a function $f: (0,1)\to \Z^+$ such that the following is true. For any $\eps
>0$, there is a one-sided tester with query complexity $f(\eps)$ that distinguishes OCF
functions from functions $\eps$-far from OCF.
\end{theorem}

In fact, odd-cycle-freeness is not just an ``example'' of a property shown to be testable
by \cite{BGS10}: {\em any} monotone
property that can be expressed as freeness from solutions to an infinite set of linear equations
is equivalent to the odd-cycle-freeness property (see Section \ref{sec:conclusion} for the
short argument). Thus, we view the study of odd-cycle-freeness as an important step
towards understanding the testability of monotone linear-invariant properties.

\ignore{As observed in  \cite{BGNX10},  \cite{BGS10} do indeed prove testability of an infinite collection of properties that were not already known to be testable.
In addition, \cite{BGS10}  propose a conjectured characterization of all  linear-invariant
properties that are testable with one-sided error,  stating that they must be described by
the intersection of  a possibly infinite collection of $(M, \sigma)$-free
properties. These  correspond in fact to those  properties that are preserved under
restrictions to subspaces, and  were called {\em subspace hereditary}.}

 Surprisingly, a similar picture   has been staged in the world of testable properties in the {\em dense graph model.}
Just as for Boolean functions,  triangle freeness, which was shown (implicitly) by
Ruzsa and Szem\'eredi~\cite{RSzeme78} to be testable,
 brought up a wealthy perspective in the area. It was followed by exciting results in
 testing $H$-freeness \cite{AlonDLRY94} and induced $H$-freeness~\cite{AFKS} which are
 somewhat analogous to the results on testing monotone and non-monotone properties
 of Boolean functions. This direction culminated with a nearly complete characterization
 of all  properties that are testable with one-sided error and constant number of
 queries~\cite{AS08}.


The proofs in \cite{BGS10} draw heavily on the ideas used in
characterizing general graph properties, such as the Szemer\'edi regularity lemma and
Ramsey theory, and just like in that case,
 the query complexity bounds obtained are a tower of exponentials.  Thus, the bound
 obtained in \cite{BGS10} for $f(\eps)$ in Theorem \ref{thm:testable} above is an
 embarrassingly weak tower of exponentials.  This brings up the question
 of characterizing the class of linear-invariant properties which can be tested with
 $\poly(1/\eps)$ queries.  In this work, we show that odd-cycle-freeness belongs to this class of
properties by the way of two different proofs, each with its own message.

\subsection{The Edge-Sampling Test}\label{sec:edge-descrip}

Our first approach for testing \ocf\ relies on reducing testing \ocf\ in Boolean functions to testing bipartiteness of the Cayley graph associated with the function. More precisely,
we will show that the following algorithm is a tester for \ocf .

\begin{center}
\fbox{
\parbox{\columnwidth} {
\underline{{\sf Edge-sampling test} (Input: oracle access to
  $f: \cube \to \zo$)}
\begin{enumerate}
\addtolength{\itemsep}{-2mm}
\item Uniformly pick $\al_1, \ldots, \al_{k} \in \F_2^n$, where $k=\tilde{O}({1}/{\eps})$. Let $G=\{\al_i-\al_j~:~ i < j\}$.
\item Accept if and only if the restriction of $f$ to $G$ is odd-cycle-free.
\end{enumerate}
}}
\end{center}


\begin{theorem}\label{thm:ocf1}
The edge-sampling test is a one-sided tester for odd-cycle-freeness with query complexity
$\tilde{O}(1/\eps^2)$.
\end{theorem}

\ignore{Notice that \ocf\ is a monotone property characterized by an infinite number of forbidden equations, and hence by the results of~\cite{BGS10} it is testable with a huge number of queries. It is easy to see that it is also closed under restrictions to subspaces, namely it is subspace hereditary. The above theorem thus provides a much better bound on its query complexity.
The bound in Theorem \ref{thm:ocf1} could be compared with the trivial lower bound of $1/\eps$ queries necessary to even observe an element in the support of the given function. \enote{hope for better bounds}
}
For a function $f: \F_2^n \to \{0,1\}$, define the Cayley graph $\calg(f)=(V, E)$ to be
the graph with vertex set $V=\F_2^n$ and edge set $E=\{(x, y)~|~ f(x-y) = 1 \}$. It's not hard to show that
$f$ is far from odd-cycle-free if and only if $\calg(f)$ is far from any bipartite {\em
  Cayley} graph. The crux in the analysis of the above testing algorithm involves showing
that if $f$ is
$\eps$-far from being odd-cycle-free, then its Cayley graph $\calg(f)$ is $\epsilon/2$-far from {\em any} bipartite graph.
The proof relies on the well-known fact that the eigenvalues of the adjacency matrix of
$\calg(f)$ are exactly the Fourier coefficients of $f$, and uses spectral techniques to
analyze the distance to bipartiteness. Finally, the test emulates the test for
bipartiteness of~\cite{AlonK02}.

In fact, we prove something stronger: the distance of $\calg(f)$ from bipartiteness is
{\em exactly} half the distance of $f$ from \ocf. Using the fact \cite{GGR,ADKK} that one can
estimate a graph's distance from bipartiteness using $\poly(1/\eps)$ queries to within
additive error $\eps$, it follows that one can estimate the distance of $f$ from \ocf.
\begin{theorem}\label{thm:distest}
There exists an algorithm that, given oracle access to a function $f: \F_2^n \to \zo$ and
a parameter $\eps \in (0,1)$, makes $\poly(1/\eps)$ queries and returns the distance of
$f$ to \ocf\ to within an error of $\pm \eps$.  The same holds for approximating
$\min_\alpha \hat{f}(\alpha)$ to within an error of $\pm \eps$.
\end{theorem}

The second statement is because of the connection between the distance to \ocf\ and Fourier
coefficients mentioned earlier. Using the above, we also obtain a $\poly(1/\eps)$-query algorithm to approximate
distance to linearity that seems different from previously known ones \cite{BLR, PRR}.

%
%

\subsection{The Subspace Restriction Test}\label{sec:canonical}

Call a tester for a graph property $\calp$ {\em canonical} if it works by picking a set of
vertices uniformly at random, querying all the edges spanned by these vertices and
checking whether the induced graph satisfies $\calp$ or
not. \cite{Alon02, GoldreichTrevisan} showed that if $\calp$ is a {\em hereditary} graph
property (i.e., if a graph $G$ satisfies $\calp$, then so does every
induced subgraph of $G$), then $\calp$ can be in fact tested using a canonical tester
with only a quadratic blowup in the query complexity.  Moreover, for many natural
hereditary graph properties, and in particular, for the property of graph bipartiteness,
there is asymptotically no loss in using a canonical tester.
The existence of a {canonical tester} also makes convenient proving lower bounds for
hereditary graph properties. It is thus natural to ask if a similar  theorem can be proved
in the context of testing properties of Boolean functions.

Suppose $\calp$ is a {\em subspace-hereditary} linear-invariant property of Boolean
functions, meaning that if a function $f: \F_2^n \to \{0,1\}$ satisfies $\calp$, then so
does $f$ restricted to any linear subspace of the domain. Subspace-hereditary properties
are especially interesting because they include most natural linear-invariant properties
and are conjectured in \cite{BGS10} to be exactly the
class of one-sided testable linear-invariant properties (modulo some technicalities).
Now, just as a canonical tester
for a hereditary graph property works by sampling a set of vertices $S$ and querying all
edges induced in $S$, one defines a canonical tester for a subspace-hereditary
property $\calp$ to be the algorithm that, on input $\eps \in (0,1)$ and
oracle access to $f: \F_2^n \to \zo$, chooses uniformly at random a linear subspace $H
\leq \F_2^n$ of dimension $d(\eps,n)$ (for some fixed function $d: [0,1] \times \Z^+ \to
\Z^+$) and accepts if and only if $f$ restricted to $H$
satisfies $\calp$.  The query complexity of the canonical tester is obviously
$2^{d(\eps,n)}$, the size of the subspace inspected by the tester. It is shown in
\cite{BGS10}, using ideas similar to those in \cite{Alon02, GoldreichTrevisan}, that any
tester for a subspace-hereditary linear-invariant property can be converted to be of
canonical form, but at the expense of
an exponential blowup in the query complexity. The question that arises then is whether,
instead of an exponential blowup, only a polynomial blowup in the query complexity is
always possible.

\begin{question}\label{ques:canpoly}
Given a subspace-hereditary\ignore{\footnote{This question can be asked for arbitrary properties
  but then the formulation becomes a bit more cumbersone. We defer this more detailed
  discussion to the full version. The reader can consult
\cite{GoldreichTrevisan} for the proper way to define a canonical tester for arbitrary
properties.}}\abnote{I removed the footnote about non-subspace-hereditary properties. Do
we really want to talk about them here?} property $\calp$ that can be tested with $q$
queries, is there always a canonical tester of complexity $\poly(q)$?
\end{question}

This seems to be a hard question in general.  However, in this work, we show that for
the property of odd-cycle-freeness, the answer to Question \ref{ques:canpoly} is
affirmative. (Note that the edge-sampling test is not canonical.)
 \begin{center}
\fbox{
\parbox{\columnwidth} {
\underline{{\sf Subspace restriction test} (Input: oracle access to
  $f: \cube \to \zo$)}
\begin{enumerate}
\addtolength{\itemsep}{-2mm}
\item Uniformly pick $\al_1, \ldots, \al_{k} \in \F_2^n$, where $k=O(\log \frac{1}{\eps})$. Let $H$
be the linear subspace spanned by $\al_1, \ldots, \al_{k}$.
\item Accept if and only if the restriction of $f$ to $H$ is odd-cycle-free.
\end{enumerate}
}}
\end{center}

\begin{theorem}\label{thm:ocf2}
The subspace restriction test is a one-sided tester for odd-cycle-freeness with query complexity
$O({1}/{\eps^{20}})$.
\end{theorem}

The analysis of the subspace-restriction test relies on a Fourier analytic argument.
One can easily see that the test accepts every OCF function.
The main insight is that certain properties of the Fourier spectrum of a function that is
$\eps$-far from being \ocf\ are preserved under random restrictions to small subspaces.

Note that Theorem \ref{thm:ocf2} implies the combinatorial fact that for any function $f$ that is $\eps$-far from
\ocf, there must exist a short witness to this fact. That is, there must exist
$x_1, \dots, x_k \in \supp(f)$ with $x_1 + \cdots + x_k = 0$ and $k = {O}(\log 1/\eps)$
an odd integer.  In fact, Theorem \ref{thm:ocf2} asserts that there must exist many such
witnesses, but a priori, it is not clear that even one such witness exists. This is in
contrast to properties such as triangle-freeness studied in \cite{Green05},
where witnesses to violations of triangle-freeness are, by definition, short, and for
testability, one ``only'' needs to show that there exist many such witnesses.

\subsection{Organization}

The rest of the paper is organized as follows. In Section \ref{sec:charac} we show that one can
relate the distance of $f$ from \ocf\ to the Fourier expansion of $f$. In Section \ref{sec:edge-sample} we use this relation
together with some results from spectral graph theory in order to analyze the edge-sampling test and thus prove
Theorem \ref{thm:ocf1}. We also show in this section that a strengthening of the analysis
for the edge-sampling test yields a distance estimator.
In Section \ref{sec:subspace}, we analyze the subspace restriction test and prove Theorem \ref{thm:ocf2}.
Finally, Section \ref{sec:conclusion} contains some concluding remarks and a discussion of some open problems related to this paper.


\section{Odd-Cycle-Freeness and the Fourier Spectrum}\label{sec:charac}

Our goal in this section is to give two reformulations of OCF, one of a geometric flavor
and one in terms of the coefficients of the Fourier expansion of $f$. These
characterizations of OCF will be useful in the analysis of both the edge-sampling test and
the subspace restriction test which will be given in later sections.
We begin by recalling some basic facts about Fourier analysis of Boolean
functions.

The orthonormal characters $\left\{\chi_{\alpha}:\F_2^n\ra \R,
  ~\chi_{\alpha}(x)=(-1)^{\alpha \cdot x} \right\}_{\alpha\in \F_2^n}$ form a basis for
the set of $\zo$-valued functions defined over $\F_2^n$, where the inner product is given
by $\langle f, g\rangle=\E_{x}[f(x)g(x)]$.  The Fourier coefficient of $f$ at $\alpha\in
\F_2^n$ is  ${\widehat f}(\al)=\E_{x}[f(x)\chi_{\al}(x)]$.
The {\em density} of $f$ is $\rho=\E_x f(x)=\widehat f(0)$, and notice that
$\rho=\max_{\alpha\in \F_2^n} |\widehat f(\al)|.$ The support of $f$ is $\supp(f)=\{x \in
\F_2^n| f(x)\not = 0\}.$  Parsevals's identity states that $\sum_{\al} {\widehat f (\alpha)}^2=\E_x
f(x)^2=\rho.$ Also, for all $\alpha$, $\hat{f}(\alpha) \geq \max(-\rho,-1/2)$.

\ignore{The Cayley graph  over $\F_2^n$ with generating set $S=\supp(f)$ is the
graph  $\calg(f)=(V, E)$ with vertex set $V=\F_2^n$, edge set $E=\{(x, y)~|~ x-y\in S \}$,
and regular degree $d=|S|$. }

We first notice that the presence of cycles in a function induces a certain distribution
of the density of the function on halfspaces.
\begin{claim}~\label{cl:ocf1}
Let $f:\F_2^n\ra \{0,1\}$. Then:
\begin{enumerate}
\item[(a)]
 $f$ is \ocf\ if and only if  there exists $\alpha\in \F_2^n$ such that for all $x \in
 \supp(f)$, $  \alpha \cdot x =1$ (i.e., there exists a linear subspace of dimension
 $n-1$ that does not contain any element of $\supp(f)$).
\item[(b)]
$f$ is $\eps$-far from \ocf\ if and only if for every $\alpha\in \F_2^n$, it holds that for
at least  $\eps 2^n$ many  $x \in \supp(f)$,~ $\alpha \cdot x=0$
(i.e., {\em every} linear subspace of dimension $n-1$ must contain at least $\eps 2^n$
elements of $\supp(f)$).
\end{enumerate}

\end{claim}
\begin{proof}
We first prove part (a). To see the ``if'' direction, suppose $f$ is not odd-cycle-free, but
there exists $\alpha$ such that for all $x \in \supp(f)$, $\alpha \cdot x = 1$. Now, let
$x_1, \dots, x_{k-1} \in \supp(f)$ be such that $x_1 + \cdots + x_{k-1} \in \supp(f)$ and
$k$ is odd. But then $\alpha \cdot (x_1 + \cdots + x_{k-1}) = \sum_{i=1}^{k-1} \alpha
\cdot x_i = 0$, a contradiction. For the opposite direction, suppose $f$ is
odd-cycle-free. If $f$ is the zero function, we are vacuously done.  Assuming otherwise,
let $S=\supp(f)$, and consider the set $H'=\{x_1 + \cdots + x_k~ :~ x_1,\dots, x_k
\in S, \text{ and } k\geq 0 \text{ is even}\}$. Since $f$ is \ocf, $H'\cap
S=\emptyset$.  It is easy to see that $H'$ is a linear subspace of codimension $1$ inside $\sspan(S)$.
It follows that $H'$ can be extended to a subspace $H$ of dimension $n-1$ such that $H\cap
S=\emptyset$.

For part (b), if $f$ is $\eps$-far from being \ocf, then by part (a), every linear
subspace $H$ of dimension $n-1$ must contain  $\eps 2^n$ elements of $\supp(f)$ (otherwise
removing less than $\eps 2^n $  points from $\supp(f)$ would create a function that is
\ocf.) The converse follows again by part (a).
\end{proof}

\begin{lemma} \label{lem:ocf2}
Let $f:\F_2^n\ra \{0,1\}$. Then
\begin{enumerate}
\item[(a)]\label{a} $f$ is \ocf\ if and only if there exists $\alpha\in \F_2^n$ such that ${\widehat f}(\alpha)=-\rho.$
\item[(b)]\label{b} $f$ is $\eps$-far from being \ocf\ if and only if for all $\beta \in
  \F_2^n$,~ $\widehat f(\beta) \geq -\rho+2\eps.$
\item[(c)] The distance of $f$ from \ocf\ is exactly $\frac{1}{2}\left(\rho + \min_{\alpha}
  \hat{f}(\alpha)\right)$.
\end{enumerate}
\end{lemma}
\begin{proof} 
By Claim~\ref{cl:ocf1}, $f$ is \ocf\ if and only if there exists $\alpha$ such that $\alpha \cdot x=1$ for all $x\in \supp(f)$, and so it follows that $$\widehat f(\alpha)=\E f(x) (-1)^{\alpha\cdot x}=-\rho.$$
This implies item (a) of the lemma. To derive item (b), we get from Claim~\ref{cl:ocf1}, that $f$ is $\eps$-far from \ocf\ if an only if any halfspace
contains at least an $\eps$ fraction of the domain, implying that
for each $\beta\in \F_2^n$:
\begin{eqnarray*}
  \widehat f(\beta) &=&\E_{x\in \F_2^n}~~ f(x) (-1)^{\beta\cdot x} = \frac 1 {2^n}\left( \sum_{x\in \supp, \beta\cdot x=0} 1+ \sum_{x\in \supp, \beta\cdot x=1} (-1)\right)\\
&\geq& \eps+(-\rho+\eps)=-\rho+2\eps.
\end{eqnarray*}
Finally, observe that item (c) is just a restatement of item (b).
\end{proof}

We see that the minimum Fourier coefficient of $f$ determines its distance
from \ocf. Since Fourier coefficients also measure correlation to linear
functions, it is natural to ask about the relationship between a function's distance
to \ocf\ and its distance to linearity\footnote{A function $f: \F_2^n \to \F_2$ is said to be linear
if $f(x+y) = f(x) + f(y)$ for all $x,y$ (the range $\{0,1\}$ has been identified with
$\F_2$). Note that linear functions are
  \ocf. However, the converse is certainly false, since the function $f(x) = x_1 x_2$ is
 \ocf\ but $1/4$-far from linear.}.
Easy Fourier analysis shows that the distance of a function $f:\F_2^n \to \F_2$ to
linearity is exactly $\min(\rho, \frac{1}{2} + \min_{\alpha} \hat{f}(\alpha))$. So,  the distance to
linearity, in contrast to \ocf, is not always determined by the minimum
Fourier coefficient.


\section{The Edge-Sampling Test}\label{sec:edge-sample}

In this section, we analyze the edge-sampling test and prove Theorem \ref{thm:ocf1}.
The analysis starts with the characterization of OCF given in the previous section and
then proceeds to reduce the problem of testing OCF for Boolean functions to testing
bipartiteness in dense graphs. We then show why Theorem \ref{thm:distest} follows.

\ignore{ The main idea is to construct a Cayley graph corresponding to the function, such that the
 graph is bipartite if and only if the function is odd-cycle-free.  Central to the
 analysis of the reduction is the fact that the eigenvalues of the Cayley graph correspond
 to the Fourier coefficients of the function. }

For a function $f: \F_2^n \to \{0,1\}$, define the Cayley graph $\calg(f)=(V, E)$ to be
the graph with vertex set $V=\F_2^n$ and edge set $E=\{(x, y)~|~ f(x-y) = 1 \}$.
Let us denote by $N = 2^n$ and let $A_{\calg}$ be the adjacency matrix of $\calg$.
The next lemma is well-known but we include its proof for the sake of completeness.
\begin{lemma}\label{lem:cayley-eigs}
For any $\alpha\in \F_2^n$, the character $\chi_{\alpha}$ is an eigenvector of $A_{\calg}$
of normalized eigenvalue $\widehat f(\alpha)$. Moreover, the set $\{ 2^n\widehat f(\alpha)
\}_{\alpha}$ is exactly the set of all the eigenvalues of $A_{\calg}$. 
\end{lemma}

\begin{proof}
Notice that the entry indexed by $x_i$ in $b=A \chi_{\alpha}$ is
\begin{eqnarray*}
b_i&=&\sum\limits_{x_j\in \F_2^n} f(x_i-x_j) \chi_{\alpha}(x_j)\\
   &=& \sum\limits_{x\in \F_2^n} f(x) \chi_{\alpha}(x_i-x)\\
   &=& \chi_{\alpha}(x_i) \sum\limits_{x\in \F_2^n} f(x) \chi_{\alpha}(x)\\
   &=& \chi_{\alpha}(x_i) (2^n \widehat f(\alpha)).
\end{eqnarray*}
Therefore, $A \chi_{\alpha}=(2^n \widehat f(\alpha)) \chi_{\alpha}$, and since the set of
characters contains $2^n$ orthogonal vectors, the lemma follows.
\end{proof}

We remind the reader that in the context of testing graph properties, a graph is
$\epsilon$-far from being bipartite if one needs to
remove at least $\epsilon n^2$ edges in order to make it odd-cycle-free. In order to be
able to apply results concerning testing odd-cycle-freeness in graphs we will have to
prove that $\calg(f)$ is in fact $\epsilon$-far from {\em any} bipartite graph.  To this
end, we show the following lemma that relates the distance to being bipartite to the least
eigenvalue of the adjacency matrix. In what follows, we denote by $e(S)$ the number of edges inside
a set of vertices $S$ in some graph $G$.

\begin{lemma}\label{lem:vertices-edges}
Let $\lambda_{min}$ be the smallest eigenvalue of the adjacency matrix $A$ of an
$n$-vertex $d$-regular graph $G$. Then for every $U\subseteq V(G)$, we have
\begin{equation}\label{eqedges}
e(U) \geq \frac{|U|}{2n} \left(|U| d+\lambda_{min} (n-|U|) \right).
\end{equation}
\end{lemma}

\begin{proof} Let $u$ be the indicator vector of $U$. We clearly have
$$
u^TAu=2e(U)\;.
$$
Since $A$ is symmetric it has a collection of eigenvectors $v_1,\ldots,v_n$ which form an orthonormal basis for $\mathbb{R}^n$.
Let $\lambda_1,\ldots,\lambda_n$ be the eigenvalues corresponding to these eigenvectors where $\lambda_n=\lambda_{min}$.
Suppose we can write $u=\sum^n_{i=1}\alpha_iv_i$ in this basis and note that $\sum^n_{i=1}\alpha^2_i=|U|$.
Since $G$ is $d$-regular, $(1/\sqrt{n},\ldots,1/\sqrt{n})$ is an eigenvector of $A$. Suppose this is $v_1$ and note that this means that $\lambda_1=d$ and $\alpha_1=|U|/\sqrt{n}$. Combining the above observations we see that
\begin{eqnarray*}
u^TAu &=& \sum^n_{i=1}\lambda_i\alpha^2_i\\
     &=& d|U|^2/n+\sum^n_{i=2}\lambda_i\alpha^2_i\\
     &\geq&  d|U|^2/n + \lambda_{min} \left(\sum^n_{i=2}\alpha^2_i\right)\\
     &=& d|U|^2/n + \lambda_{min} (|U|-|U|^2/n)\;.
\end{eqnarray*}
We now get (\ref{eqedges}) by combining the above two expressions for $u^TAu$.
\end{proof}

\begin{corollary}\label{cor:dense-halves}
Let $G$ be an $n$-vertex $d$-regular graph with $\lambda_{min} \geq -d +2\epsilon n$. Then $G$ is $\epsilon/2$-far from being bipartite.
\end{corollary}
\begin{proof} It is clearly enough to show that in any bipartition of the vertices of $G$ into sets $A,B$ we have $e(A)+e(B)\geq \frac12\epsilon n^2$.
So let $(A,B)$ be one such bipartition and suppose $|A|=cn$ and $|B|=(1-c)n$. From Lemma \ref{lem:vertices-edges} we get that
\begin{eqnarray*}
e(A) &\geq& \frac{c}{2}(dcn+(-d+2\epsilon n)(n-cn))\\
     &=&\frac{c}{2}(2\epsilon n^2 - dn)+\frac{c^2}{2}(2dn-2\epsilon n^2)\;,
\end{eqnarray*}
and similarly
$$
e(B) \geq \frac{1-c}{2}(2\epsilon n^2 - dn)+\frac{(1-c)^2}{2}(2dn-2\epsilon n^2)\;.
$$
Hence
\begin{eqnarray*}
e(A)+e(B) &\geq& \frac12(2\epsilon n^2-dn)+\frac12(c^2+(1-c)^2)(2dn-2\epsilon n^2)\\
          &\geq& \frac12(2\epsilon n^2-dn)+\frac14(2dn-2\epsilon n^2)\\
          &=& \frac12\epsilon n^2\;,
\end{eqnarray*}
where in the second inequality we use the fact that $c^2+(1-c)^2$ is minimized when $c=1/2$.
\end{proof}


We can now derive the following {\em exact} relation between the \ocf\ property of functions and the bipartiteness of the corresponding Cayley graphs.

\begin{corollary}\label{lem:main}
Let $f:\F_2^n\ra \F_2$. If $f$ is $\eps$-far from being \ocf, then $\calg(f)$ is $\eps/2$-far from being bipartite.
Furthermore, if $f$ is $\eps$-close to being \ocf, then $\calg(f)$ is $\eps/2$-close to being bipartite.
\end{corollary}
\begin{proof}
Suppose $f$ is $\eps$-far from being \ocf. Let us suppose $\supp(f) = \rho N$ where $N = 2^n$.  By Lemmas \ref{lem:ocf2} and
\ref{lem:cayley-eigs}, if $f$ is $\eps$-far from being \ocf, then the smallest eigenvalue
of the adjacency matrix of $\calg(f)$ is $\lambda_{min} \geq -(\rho+2\eps)N$.   For a function
$f$, recall that $\calg(f)$ is a regular graph with degree $d = |\supp(f)| = \rho N$.
Hence, we can use Corollary \ref{cor:dense-halves} to infer that $\calg(f)$ must be $\epsilon/2$-far from being
bipartite.

Suppose now that $f$ is $\eps$-close to being \ocf, and let $S$ be the set of $\epsilon N$
points in $\F^n_2$ whose removal from the support of $f$ makes it
OCF. Call this new function $f'$. Observe that every $x \in \F^n_2$ for which $f(x)=1$
accounts for $N/2$ edges in $\calg(f)$. Hence, removing from $\calg(f)$ all edges
corresponding to $S$ results in the removal of at most $\frac12 \epsilon N^2$ edges. To finish the
proof we just need to show that the new graph $\calg(f')$ (note that the new graph is
indeed the Cayley graph of the new function $f'$)
does not contain any odd cycle.
Suppose to the contrary that it contains an odd cycle $\al_1,\ldots,\al_k,\al_1$.
For $1 \leq i \leq k$ set $x_i=\al_{i+1}-\al_i$. Then by definition of $\calg(f')$ we have $f'(x_1)=\ldots=f'(x_k)=1$. Furthermore, as
$\al_1+\sum^k_{i=1}x_i=\al_1$ (since we have a cycle in $\calg(f')$) we get that $\sum^k_{i=1}x_i=0$ so $x_1,\ldots,x_k$ in an odd-cycle in $f'$ contradicting the assumption
that $f'$ is OCF.
\end{proof}

We are now ready to complete the proof of Theorem \ref{thm:ocf1} using the following
result of Alon and Krivelevich~\cite{AlonK02}.

\begin{theorem}[\cite{AlonK02}]\label{thm:AK}
Suppose a graph $G$ is $\epsilon$-far from being bipartite. Then a random subset of
vertices of $V(G)$ of size $\tilde O(1/{\eps})$ spans a non-bipartite
graph with probability at least $3/4$.
\end{theorem}

\begin{proof-of-theorem}{\ref{thm:ocf1}}
%
First, if $f$ is OCF then the test will clearly accept $f$ (with probability 1). Suppose now that $f$ is $\epsilon$-far from being OCF. Then by
Corollary \ref{lem:main} we get that $\calg(f)$ is $\eps/2$-far from being bipartite. Now notice that we can think
of the points $\al_1,\ldots,\al_k \in \F^n_2$ sampled by the edge-sampling test as vertices sampled from $\calg(f)$.
By Theorem \ref{thm:AK}, with probability at least $3/4$, the vertices $\al_1,\ldots,\al_k$ span an odd-cycle of $\calg(f)$.
We claim that if this event happens, then the edge-sampling test will find an odd-cycle in $f$. Indeed, if $\al_1,\ldots,\al_k,\al_1$ is an odd-cycle in $\calg(f)$, then
as in the proof of Corollary \ref{lem:main} this means that $\al_{2}-\al_1,\ldots,\al_{3}-\al_2,\ldots,\al_{1}-\al_k$ is an odd-cycle of $f$.
Finally, the edge-sampling test will find this odd cycle, since it queries $f$ on all points $\al_i-\al_j$.
\end{proof-of-theorem}

In order to obtain Theorem \ref{thm:distest}, observe
that by Corollary \ref{lem:main}, the distance to OCF for
a function $f$ is exactly double the distance to bipartiteness for the graph
$\mathcal{G}(f)$.  We now invoke the following result of Alon, de la Vega, Kannan and
Karpinski \cite{ADKK}, which improved upon a previous result of Goldreich, Goldwasser and Ron \cite{GGR}.
\begin{theorem}[\cite{ADKK}]\label{thm:ADKK}
For every $\eps > 0$, there exists an algorithm that, given input graph $G$, inspects a
random subgraph of $G$ on $\tilde{O}(1/\eps^4)$ vertices and estimates the distance from $G$ to
bipartiteness to within an additive error of $\eps$.
\end{theorem}

\begin{proof-of-theorem}{\ref{thm:distest}}
Combining Theorem~\ref{thm:ADKK} with Lemma \ref{lem:main}
we immediately obtain a $\poly(1/\eps)$ query algorithm that estimates the distance to
odd-cycle-freeness with additive error at most $\eps$. 
  Since one can use sampling to estimate $\rho$ to within an additive error $\eps$ using $\poly(1/\eps)$ queries, it follows
from item (c) of Lemma \ref{lem:ocf2} that one can estimate $\min_\alpha \hat{f}(\alpha)$ to
within an additive error of $\eps$ using $\poly(1/\eps)$ queries.
\end{proof-of-theorem}

As we have mentioned earlier, the distance of $f$ from being linear is given by
$\min(\rho, \frac{1}{2} + \min_{\alpha} \hat{f}(\alpha))$, where $\rho$ is the density of $f$.
Therefore, given an estimate of $\rho$ and $\min_\alpha \hat{f}(\alpha)$ for some
function $f: \F_2^n \to \F_2$, one can also estimate the distance of $f$ to linearity.
Theorem \ref{thm:distest} thus gives a new distance estimator for linearity, and hence also a two-sided
tester for the property of linearity, both with $\poly(1/\eps)$ query complexity.

\ignore{
We also note here that one can easily obtain a $\tilde{O}(n)$-query distance estimator for
odd-cycle-freeness by using the Goldreich-Levin \cite{GoldreichL89} algorithm. For any
fixed $\eps_1$, we can
test if the distance is at least $\eps_1$ by using Goldreich-Levin to obtain the list of
Fourier coefficients that are of absolute value at least $|\rho-2\eps_1|$. This list is of
constant\footnote{We can always move $\eps_1$ by half the
additive error of the estimator if $\rho-2\eps_1$ is too close to zero.} length (for fixed
error parameter). Then, by sampling, we can estimate each coefficient in this list upto
small additive error and determine if any of them are smaller than $-\rho +2\eps_1$ in
(signed) value. Similarly, we can also test if the distance is at most $\eps_2$. Finally,
we can use binary search to estimate the distance to within an additive error.  The query
complexity of this algorithm is dominated by that of the Goldreich-Levin algorithm, $O(n
\log n)$. Note that we can improve the query complexity to a constant because for
estimating the distance, it is not necessary to explicitly obtain the list of all
coefficients of large magnitude.}

\section{The Subspace Restriction Test}\label{sec:subspace}

In this section, we analyze the subspace restriction test and prove
Theorem~\ref{thm:ocf2}.
We start with a few notational remarks. For $f:\F_2^n\ra \{0,1\}$ and subspace $H \leq
\F_2^n$, let $f_H: H\ra \{0,1\}$ be the restriction of $f$ to $H$, and let $\rho_H$ denote the
density of $f_H$, namely $\rho_H=\Pr_{x\in H}[ f_H(x)=1]$.
 For $\alpha\in \F_2^n$  and subspace $H$, define the restriction of the
 Fourier coefficients of $f$ to a subspace $H$ to be
\[
{\widehat f}_H(\alpha) =\E_{x\in H}[f(x)\chi_{\alpha}(x)].
\]
Recall that the character group of $H$ is isomorphic to $H$ itself, and
so, $f_H=\sum_{\alpha \in H} \widehat{f}_H(\alpha) \chi_\alpha$.  \enote{do we need to make this \\more formal?}
The dual of $H$ is the subspace $H^{\perp}=\{x\in \F_2^n~|~ \langle x, a\rangle =0 ~\forall a\in H\}$. Note that
 $\widehat{f}_H(\alpha) = \widehat{f}_H(\beta)$
  whenever $\alpha\in \beta+H^{\perp}.$
  The convolution of $f, g:\F_2^n\ra \F_2$ is ${f*g}:\F_2^n\ra \F_2$, $({f*g})(c)=\E_{x\in \F_2^n} f(x+c)f(x).$ It is known that $\widehat{f*g}={\widehat f}\cdot{\widehat g}.$
In what follows, we will let $h$ be the size of the the subspace $H$.

The strategy of the proof is to use Lemma \ref{lem:ocf2} and reduce the analysis to showing
that if every nonzero Fourier coefficient of $f$ is at least $-\rho +
2\eps$, then  for a random linear subspace $H$, with probability
 $2/3$, every nonzero Fourier coefficient of $f_H$ is strictly
greater than $-\rho_{{H}}$ (where, again, $f_H : H \to \{0,1\}$ is the
restriction of $f$ to $H$ and $\rho_H$ is the density of $f_H$).

A useful insight into why this should be true is that the restricted Fourier coefficients
are  concentrated around the non-restricted counterparts, deviating from them by an amount
essentially inversely proportional with the size of the subspace $H$. A direct union bound
type argument is however too weak to give anything interesting when the size of $H$ is
small. The idea of our proof is to separately analyze the restrictions of the large and
small coefficients.
Understanding the restrictions of the small coefficients is the more difficult part of the
argument, and the crux of the proof relies on noticing that the moments of the Fourier
coefficients are  also preserved under restrictions to subspaces. In particular, an
analysis of the deviation of the fourth moment implies that one can balance the parameters
involved so that even when $H$ is of size only $\poly(1/\eps)$, no restriction of the
coefficients of low magnitude can become as small as $-\rho_H$.


We first show that the restriction of $f$ to a random linear
subspace does not change an individual Fourier coefficient by more than a small
additive term dependent on the size of the subspace. This follows from standard Chebyshev-type concentration bounds.
\begin{lemma}\label{lem:coeffdev}
\[
\Pr_H\left[ \left|{\widehat f_H}(\alpha)- {\widehat f}(\alpha)\right| \geq \frac{2}{h}+\eta \right]\leq \frac{14}{h\eta^2}.
\]
\end{lemma}

\begin{proof}
\ignore{
First, observe that if $K \leq \F_2^n$ is a subspace of
size $k$, then for any $\alpha \in \F_2^n$, $\E_{x \in
  K}[f(x)\chi_\alpha(x)] =\frac{1}{k}f(0)+(1-\frac{1}{k})\E_{x \in
  K-\{0\}}[f(x)\chi_\alpha(x)]$. Hence, $\E_{x \in
  K}[f(x)\chi_\alpha(x)] \leq \E_{x \in
  K-\{0\}}[f(x)\chi_\alpha(x)] + \frac{1}{k}$ and $\E_{x \in
  K}[f(x)\chi_\alpha(x)] \geq \left(1-\frac{1}{k}\right) \E_{x \in
  K-\{0\}}[f(x)\chi_\alpha(x)] \geq \E_{x \in
  K-\{0\}}[f(x)\chi_\alpha(x)] - \frac{1}{k}$.
}
\ignore{$\E_{x \in
  K}[f(x)\chi_\alpha(x)] \leq 1/k +  \E_{x \in
  K-\{0\}}[f(x)\chi_\alpha(x)]$ and $\E_{x \in
  K}[f(x)\chi_\alpha(x)] \geq -1/k + (1-1/k)\E_{x \in
  K-\{0\}}[f(x)\chi_\alpha(x)] \geq -2/k + \E_{x \in
  K-\{0\}}[f(x)\chi_\alpha(x)]$.
  }

Now, consider the deviation of $\E_H[\widehat{f}_H(\alpha)]$ from $\widehat{f}(\alpha)$:
\begin{align*}
 \E_H[{\widehat f}_H(\alpha)]
&= \E_H \E_{x\in H}[f(x)\chi_{\alpha}(x)]\\
&\geq \E_H \left(\left(1 - \frac{1}{h}\right)\E_{x\in H-\{0\}}[f(x)\chi_{\alpha}(x)]\right)\\
&\geq \E_{x \in \F_2^n -\{0\}}[f(x) \chi_\alpha(x)] - \frac{1}{h}\\
&\geq \widehat{f}(\alpha) -\frac{1}{2^n} - \frac{1}{h}
\geq \widehat{f}(\alpha) - \frac{2}{h}
\end{align*}
Similarly:
\begin{align*}
 \E_H[{\widehat f}_H(\alpha)] \leq \widehat{f}(\alpha) + \frac{2}{h}
\end{align*}
So, it suffices to show that $\Pr\left[|\widehat{f}_H(\alpha) - \E_H
  \widehat{f}_H(\alpha)| \geq \eta\right] \leq 10/(h\eta^2)$.  We prove this
by bounding the variance of $\widehat{f}_H(\alpha)$.

\ignore{\allowdisplaybreaks
\begin{align*}
\E\left[\widehat{f}_H(\alpha)^2\right]
&= \E_H \left[\left(\E_{x\in H} f(x) \chi_\alpha(x)\right)^2\right] \\
&= \E_H \left[\left(\sum_{\beta} \widehat{f}(\beta) \E_{x \in H} \chi_{\alpha +
      \beta}(x)\right)^2\right]\\
&= \E_H \left[\left(\sum_{\beta} \widehat{f}(\beta) \mathbbm{1}_{\alpha+\beta \in H^\perp}\right)^2\right]\\
&=\sum_{\beta_1,\beta_2}\widehat{f}(\beta_1)\widehat{f}(\beta_2)
  \E_H\left[\mathbbm{1}_{\alpha+\beta_1 \in H^\perp} \mathbbm{1}_{\alpha+\beta_2 \in
      H^\perp}\right]\\
&\leq \widehat{f}^2(\alpha)  +
\frac{1}{h} \sum_\beta \widehat{f}^2(\beta) +
\frac{\widehat{f}(\alp}{h}

\sum_{{\beta_1,
    \beta_2:}\atop {|\{\alpha,\beta_1,\beta_2\}|\geq 2}}
\widehat{f}(\beta_1)\widehat{f}(\beta_2)   \E_H\left[\mathbbm{1}_{\alpha+\beta_1 \in H^\perp}
  \mathbbm{1}_{\alpha+\beta_2 \in
      H^\perp}\right] \\
&\leq \widehat{f}^2(\alpha) + \frac{1}{h} + \frac{1}{h}
\end{align*}
}
{
\begin{align*}
\E[{\widehat f}_H(\alpha)^2]
&=\E_H\left[\left(\E_{x\in H}f(x)\chi_{\alpha}(x)\right)^2\right]\\
&= \E_H\left[ \E_{x, y \in H} f(x) f(y) \chi_{\alpha}(x)\chi_{\alpha}(y)\right]\\
&\leq \E_H\left[ \Pr[\dim(\sspan(x,y)) < 2] +\E_{x, y\in H \atop {\dim(\sspan(x,y)) = 2}}
  f(x) f(y) \chi_{\alpha}(x)\chi_{\alpha}(y)\right]\\
&\leq \frac{3}{h}+\E_{x, y\in \F_2^n \atop {\dim(\sspan(x,y)) =
    2}}f(x)f(y)\chi_{\alpha}(x)\chi_\alpha(y)\\
&\leq \frac{3}{h}+ \frac{3}{2^n} + \widehat{f}^2(\alpha)\\
&\leq \frac{6}{h} + \left(\E_H \widehat{f}_H(\alpha) + \frac{2}{h}\right)^2
\leq \frac{14}{h} + \left(\E_H \widehat{f}_H(\alpha)\right)^2
\end{align*}}
Hence $\Var[{\widehat f}_H(\alpha)]\leq \frac{14}{h}$, and  the lemma now follows by
Chebyshev's inequality.
\end{proof}


As it was the case with the restricted coefficients, it can also be shown using a
straightforward variance calculation that the fourth moment
is preserved up to small additive error upon restriction to a random
$H$, when $h$ is large enough.
 For that purpose, define $A$ and $A_H$ as follows:
\begin{align*}
A~ &\eqdef~ \sum_{\alpha \in \F_2^n}\widehat{f}^4(\alpha) =  \E_{x_1, x_2, x_3 \in \F_2^n}[f(x_1)
f(x_2) f(x_3) f(x_1 + x_2 + x_3)]\\
A_H~ &\eqdef~ \sum_{\alpha \in H} \widehat{f}_H^4(\alpha) = \E_{x_1, x_2, x_3 \in H}[f(x_1) f(x_2) f(x_3) f(x_1 + x_2 +
x_3)]
\end{align*}
\ignore{
The second equality on each line follows from  using the standard fact
that $\widehat{f * f}(\alpha) = \widehat{f}^2(\alpha)$ and thus observing that $A = (f * f * f *
f)(0)$ and $A_H = (f_H * f_H * f_H * f_H)(0)$.}

Then, we have:
\begin{lemma}\label{lem:momentdev}
$$\Pr_H\left[ |A_H- A|\geq \frac{16}{h}+\eta\right]\leq \frac {500}{h\eta^2}.$$
\end{lemma}

\begin{proof}
As in the proof of Lemma \ref{lem:coeffdev}, our strategy will be to
first show that $\E_H[A_H]$ is likely to be close to $A$ and then to
bound the variance of $A_H$.
\ignore{
We proceed to bound the variance of $A_H$ over all subspaces $H$  of
size $h>2$ and start by computing $\E_H(A_H)=\E_H \E_{c\in C_l\cap H}
F(c)$.
Notice that if $c=(c_1, \ldots, c_{l-1}, \sum c_i)\in C_l^*\cap H$  then picking $H$ at random and then picking $c\in C_l\cap H$ is the same as picking $c\in (\F_2^n)^*$.
}
\begin{claim}
$$|A-\E_H[A_H]|\leq \frac{16}{h}.$$
\end{claim}
\begin{proof}
{\allowdisplaybreaks
\begin{align*}
\E_H[A_H]
&= \E_H \E_{x_1, x_2, x_3 \in H} f(x_1) f(x_2) f(x_3) f(x_1 + x_2 + x_3)\\
&\geq \E_H \left[\left(1-\frac{8}{h}\right) \E_{\substack{x_1, x_2,
x_3 \in H \\ \dim(\sspan(x_1,x_2,x_3)) = 3}} f(x_1) f(x_2) f(x_3)
  f(x_1 + x_2 + x_3) \right]\\
&\geq \E_H\left[ \E_{\substack{x_1, x_2,
x_3 \in H \\ \dim(\sspan(x_1,x_2,x_3)) = 3}} f(x_1) f(x_2) f(x_3) f(x_1 + x_2 + x_3) -
  \frac{8}{h} \right]\\
&= \E_{\substack{x_1, x_2,
x_3 \in \mathbb{F}_2^n \\ \dim(\sspan(x_1,x_2,x_3)) = 3}}  f(x_1) f(x_2)
f(x_3) f(x_1 + x_2 + x_3) - \frac{8}{h} \\
&\geq \E_{x_1, x_2, x_3 \in \F_2^n}  f(x_1) f(x_2) f(x_3) f(x_1 + x_2
+ x_3) - \frac{8}{2^n} -
\frac{8}{h}
\geq A - \frac{16}{h}
\end{align*}
and similarly:
\begin{align*}
\E_H[A_H] \leq A + \frac{16}{h}
\end{align*}}
\ignore{\begin{eqnarray*}
\E_H[A_H]&=&\E_H \E_{c\in C_l\cap H} F(c)
         = \E_{c\in C_l^*}F(c)+ \E_H \Pr_{c\in C_l\cap H} [c\notin C_l^*\cap H] F(c)\\
         &\leq&  \frac{1}{{(2^n-1)}^{l-1}}( 2^{n(l-1)} \E_{c\in C_l}F(c)-  \sum_{c\not\in C_l^*}F(c))+\frac{2^{l-1}}{h^{l-1}}\\
         &\leq& (1+\frac{1}{2^n-1})^{l-1}\E_{c\in C_l}F(c)+\frac{2}{h}\leq (1+ \frac{1}{(2^n-1)^{l-1}}+ \frac{2^{l-1}}{2^n-1})\E_{c\in C_l}F(c) + \frac{2}{h} \\
        & \leq& \E_{c\in C_l}F(c)+\frac{2^l}{2^n-1}+\frac{2}{h}
        \leq A +\frac{2^{l+1}}{h}.
\end{eqnarray*}}
\end{proof}
\begin{claim}
\[
\Var[A_H]\leq  \frac{500}{h}.
\]
\end{claim}

\begin{proof}

\ignore{In what follows we use the fact that if $\mspan(x_1, x_2, x_3) \cap \mspan(y_1, y_2, y_3)=\{0\}$ then for a randomly chosen $H$, the event that
$\mspan(x_1, x_2, x_3)\subseteq H$ and the event that $\mspan(y_1, y_2, y_3)\subseteq H$ are pairwise independent.}

\begin{align*}
\E_H[A_H^2]
&= \E_H \left[ \E_{x_1,x_2,x_3\in H \atop y_1, y_2, y_3 \in H} f(x_1) f(x_2) f(x_3)
  f(x_1 + x_2 + x_3)f(y_1) f(y_2) f(y_3)
  f(y_1 + y_2 + y_3) \right]\\
&\leq \E_H \left[\Pr[\dim(\sspan(\{x_1,x_2,x_3,y_1,y_2,y_3\}))<6]\right.\\
&~~~~~~~~\left.+ \E_{x_1,x_2,x_3,y_1,y_2,y_3 \in H
    \atop \dim(\sspan(\{x_1,x_2,x_3,y_1,y_2,y_3\}))=6}[f(x_1) f(x_2) f(x_3)
  f(x_1 + x_2 + x_3)f(y_1) f(y_2) f(y_3) f(y_1+y_2+y_3)]\right]\\
&\leq \frac{64}{h} + \E_{x_1,x_2,x_3,y_1,y_2,y_3 \in \F_2^n
    \atop \dim(\sspan(\{x_1,x_2,x_3,y_1,y_2,y_3\}))=6}[f(x_1) f(x_2) f(x_3)
  f(x_1 + x_2 + x_3)f(y_1) f(y_2) f(y_3) f(y_1+y_2+y_3)]\\
&\leq \frac{64}{h} + \frac{64}{2^n}  + A^2\\
&\leq \frac{128}{h} + \left(\frac{16}{h} + \E_H[A_H]\right)^2
\leq \frac{500}{h} + \E_H[A_H]^2.
\end{align*}

\end{proof}

The lemma now follows by Chebyshev's inequality.
\end{proof}

\ignore{
\begin{proof}[Proof of Claim~\ref{clm:lowmomentdev}]
We have that
\begin{eqnarray*}
\left|\sum_{\alpha \in L_H} {\widehat f}_H^4(\alpha)-\sum_{\alpha \in L} {\widehat f}^4(\alpha)\right|
&\leq & \left|\sum_{\alpha \in L} {\widehat f}_H^4(\alpha)-\sum_{\alpha \in L} {\widehat f}^4(\alpha)\right|\\
&\leq& \sum_{\alpha \in L} \left|{\widehat f}_H^4(\alpha)- {\widehat f}^4(\alpha) \right|\\
&\leq& \sum_{\alpha \in L} | {\widehat f}_H(\alpha)- {\widehat f}(\alpha) | ( \sum_{i=0}^{3} {\widehat
  f}_H(\alpha)^i {\widehat f}(\alpha)^{3-i})|\\
&\leq& 4 \cdot|L|\cdot \max | {\widehat f}_H(\alpha)-{\widehat f}(\alpha)|\\
&\leq& \frac{4}{\gamma^2} \left( \frac{3}{h}+\eta_1\right).
 \end{eqnarray*}

It follows that
\begin{eqnarray*}
\left|\sum_{\alpha \in S_H} {\widehat f}_H^4(\alpha)- \sum_{\alpha \in S} {\widehat f}^4(\alpha)\right|
&\leq& \left|\left(A_H-\sum_{\alpha \in L_H} {\widehat f}_H^4(\alpha)\right )-\left(A-\sum_{\alpha \in L} {\widehat f}^4(\alpha)\right)\right|\\
 &\leq& |A_H-A| + \left|\sum_{\alpha \in L_H} {\widehat f}_H^4(\alpha)-\sum_{\alpha \in L} {\widehat f}^4(\alpha)\right|\\
&\leq& \eta_2+\frac{6}{h}+ \frac{4}{\gamma^2} \left( \frac{3}{h}+\eta_1\right).
\end{eqnarray*}

\end{proof}}

Using Lemma~\ref{lem:coeffdev} and~\ref{lem:momentdev} we can now proceed with the proof of Theorem~\ref{thm:ocf2}.
\\\\

\begin{proof-of-theorem}{\ref{thm:ocf2}}
\ignore{Assume that $f$ is $\eps$-far from ocf. We show that with probability at least $2/3$ over random
choice of $H$, it is the case that none of the Fourier coefficients on $H$  deviates too much from
the corresponding coefficients of $f$, which will imply that there is no coefficient with value
$-\rho_H$, and so there must be some odd cycle in $H$. }
If $f$ is $\eps$-far from odd-cycle-free, then  $\rho > \eps$, and  by Lemma \ref{lem:ocf2}, all its Fourier
coefficients are $> -\rho+2\epsilon$.  We need to show that with
constant probability over random choice of $H$, each Fourier coefficient of $f_H$ is  $>-\rho_{H}$.
We separate these coefficients into the sets of large and small coefficients and analyze them separately.
Define $$L~\eqdef~\{\alpha ~|~ | \widehat f(\alpha)|\geq \gamma  \} \subseteq \F_2^n \mbox{\ \ \ and\ \ \  } S~\eqdef~ \F_2^n\backslash L$$
for some $\gamma < \rho$ to be chosen later. Notice that $0 \in L$. Also, by Parseval's identity, $|L|\leq 1/ \gamma^2.$
Let $L_H\subseteq H$ be the set of elements $\beta \in H$ such that there exists
$\alpha \in L$ with $\beta\in \alpha+H^{\perp},$ that is, $\beta$ is the ``projection'' of some large coefficient.
 Then $|L_H|\leq |L|.$  Let $S_H = H \backslash L_H$ be the complement of $L_H$ in $H$. \enote{H or the group of characters?}

  From Lemma \ref{lem:coeffdev}, for each $\alpha\in L$ and for any $\eta_1 \in (0,1)$, we
  have $\Pr_H\left[ |{\widehat f_H}(\alpha)- {\widehat
    f}(\alpha)| \geq   \frac{2}{h}+\eta_1 \right]\leq \frac{14}{h\eta_1^2}.$
By a union bound, with probability $1-\frac{1}{ \gamma^2} \frac{14}{h\eta_1^2}$, for every $\alpha
\in L_H$, it holds $\widehat f_H(\alpha) > \widehat f(\alpha) -\frac 2 h-\eta_1$.  Moreover, since $0\in L$,
we know $|\rho_{H}-\rho_f|\leq \frac{2}{h}+\eta_1$.  If $2\eta_1+ \frac 4 h <  2\eps$, then  for any
$\alpha \in L_H$, we have
$$\widehat
f_H(\alpha)>\widehat f (\alpha)-\frac 2 h- \eta_1>-\rho+2\eps-\frac 2 h- \eta_1>-\rho_{H} + 2\eps-\frac 4 h-2 \eta_1 > -\rho_H$$
 with probability at least $1-\frac{14}{h\gamma^2\eta_1^2}$.


We now analyze the coefficients $\beta\in S_H$ and again show that with constant probability, no
$\widehat{f}_H(\beta)$ becomes as small as $-\rho_{H}.$  As we described in the informal proof sketch
earlier, for this, we will want to analyze the fourth moment of the Fourier coefficients.

To this end, first observe that for any two Fourier coefficients $\alpha,\alpha' \in L$,
their projections are  identical if $\alpha - \alpha' \in H^{\perp}$. Over the random
choice of $H$, this happens with probability at most $\frac{1}{h}$.  Therefore, using a
union bound, we conclude that with probability at least $1- |L|^2/h = 1-\frac{1}{\gamma^4
  h}$, all the large Fourier coefficients project to distinct coefficients in $H$, namely
$|L_{H}| = |L|$.
Let us condition on this event that no two large Fourier coefficients in $L$ project to the same restricted coefficient.

Let us also condition on the event that  $|A-A_H| < \frac{16}{h}+ \eta_2$ for some $\eta_2$ to be
specified later.   Also, condition on the event that for all $\alpha\in L$, $ |{\widehat f_H}(\alpha)-
{\widehat f}(\alpha)| < \frac{2}{h}+\eta_1 $.   All of these events occur with probability at least $1
- \frac{500}{h \eta_2^2} - \frac{14}{h\gamma^2 \eta_1^2} - \frac{1}{\gamma^4 h}$ by Lemmas \ref{lem:coeffdev} and
\ref{lem:momentdev}.

\ignore{
For that we will
analyze the fraction of $l$-cycles on $H$, where $l$ will be a small even constant, say $l=4$ .}

The following claim shows that the fourth moment of the small Fourier coefficients is also preserved under a random subspace restriction.

\begin{claim}\label{clm:lowmomentdev}
$$\left|\sum_{\alpha \in  S_H} {\widehat f}_H^4(\alpha)- \sum_{\alpha  \in S} {\widehat
    f}^4(\alpha)\right|\leq \eta_2+\frac{16}{h}+ \frac{4}{\gamma^2} \left( \frac{2}{h}+\eta_1\right).$$

\end{claim}
\begin{proof}
We have that
\begin{eqnarray*}
\left|\sum_{\alpha \in L_H} {\widehat f}_H^4(\alpha)-\sum_{\alpha \in L} {\widehat f}^4(\alpha)\right|
&= & \left|\sum_{\alpha \in L} {\widehat f}_H^4(\alpha)-\sum_{\alpha \in L} {\widehat f}^4(\alpha)\right|\\
&\leq& \sum_{\alpha \in L} \left|{\widehat f}_H^4(\alpha)- {\widehat f}^4(\alpha) \right|\\
&\leq& \sum_{\alpha \in L} | {\widehat f}_H(\alpha)- {\widehat f}(\alpha) | ( \sum_{i=0}^{3} {\widehat
  f}_H(\alpha)^i {\widehat f}(\alpha)^{3-i})|\\
&\leq& 4 \cdot|L|\cdot \max | {\widehat f}_H(\alpha)-{\widehat f}(\alpha)|\\
&\leq& \frac{4}{\gamma^2} \left( \frac{2}{h}+\eta_1\right).
 \end{eqnarray*}

It follows that
\begin{eqnarray*}
\left|\sum_{\alpha \in S_H} {\widehat f}_H^4(\alpha)- \sum_{\alpha \in S} {\widehat f}^4(\alpha)\right|
&\leq& \left|\left(A_H-\sum_{\alpha \in L_H} {\widehat f}_H^4(\alpha)\right )-\left(A-\sum_{\alpha \in L} {\widehat f}^4(\alpha)\right)\right|\\
 &\leq& |A_H-A| + \left|\sum_{\alpha \in L_H} {\widehat f}_H^4(\alpha)-\sum_{\alpha \in L} {\widehat f}^4(\alpha)\right|\\
&\leq& \eta_2+\frac{16}{h}+ \frac{4}{\gamma^2} \left( \frac{2}{h}+\eta_1\right).
\end{eqnarray*}

\end{proof}

Now, on the one hand, we have: $\sum_{\alpha \in S} \widehat{f}^4(\alpha) < \gamma^2 \sum_\alpha
\widehat{f}^2(\alpha) \leq \gamma^2$.  On the other hand, $\max_{\alpha \in S_H} \widehat{f}_H^4(\alpha) \leq
\sum_{\alpha \in S_H} \widehat{f}_H^4(\alpha)$.  Therefore, combining and using Claim
\ref{clm:lowmomentdev}, we have:
\begin{equation*}
\max_{\alpha \in S_H} \widehat{f}_H^4(\alpha) < \gamma^2 + \eta_2 + \frac{16}{h} +
\frac{4}{\gamma^2}\left(\frac{2}{h} + \eta_1\right)
\end{equation*}
We need to choose the parameters such that $\max_{\alpha \in S_H} |{\widehat f}_H(\alpha)| <
\rho_{H}$, and so, it is enough to have:
\begin{equation*}
\gamma^2 + \eta_2 + \frac{16}{h} + \frac{4}{\gamma^2} \left(\frac{2}{h} + \eta_1\right) < \left(\eps -
  \frac{2}{h} - \eta_1\right)^4
\end{equation*}
Additionally, we need to ensure that the events we have conditioned on occur with probability at
least $2/3$.  So, we want:
\begin{equation*}
\frac{500}{h\eta_2^2} +\frac{14}{h \gamma^2 \eta_1^2} + \frac{1}{\gamma^4 h} < \frac{1}{3}
\end{equation*}

One can check now that the following setting of parameters satisfies both of the above constraints:
$\gamma= \eps^2/100$, $h=(10/\eps)^{20}$, $\eta_1=(\eps/10)^8$, $\eta_2=(\eps/10)^4.$
\end{proof-of-theorem}

\section{Concluding Remarks and Open Problems}\label{sec:conclusion}

\begin{itemize}

\item The main open question raised here  (Question~\ref{ques:canpoly})  is
whether it is possible in general to obtain canonical testers for subspace-hereditary
properties with only a polynomial blow up in the query complexity.  Here, we show this to
be true for OCF, and \cite{BX10} showed the existence of a canonical tester with quadratic
blowup for the triangle-freeness property. On the other hand, there is some evidence to
the contrary also. Goldreich and Ron in \cite{GR11} proved a nontrivial gap between
canonical and non-canonical testers for graph properties. They showed that there exist hereditary
graph properties that can be tested using $\tilde{O}(\eps^{-1})$ queries but for which the
canonical tester requires $\tilde{\Omega}(\eps^{-3/2})$ queries.
Perhaps, this indicates that for subspace-hereditary properties also, there is a non-trivial, maybe even
super-polynomial in this case, gap between non-canonical and canonical testers.

\item As previously mentioned, \ocf\ is in fact the only monotone property characterized by
freeness from an infinite number of equations (of rank $1$). We briefly comment here on the equivalence
between all these properties.
It is easy to see that even-length equations can be handled trivially.
Suppose now that $\calp$ is  defined by freeness from all equations  of length belonging
to the infinite set of odd integers $S=\{k_1, k_2, \ldots\}$. Note that $\mbox{OCF}\subseteq
\calp$. Now suppose $k\not\in S$ and $k$ is odd, and let $k'$ be the smallest element of
$S$ such that $k\leq k'$. If $f\in \calp$ is not free of solutions to the length $k$
equation, then $f$ is not free of solutions to the  equation of length $k'$, since a
solution $(x_1, \ldots, x_k)$ to the former induces a solution $(x_1, \ldots, x_k, x_1,
x_1\ldots, x_1)$ to the latter.

\item Another open problem that arises is to characterize the
class of linear-invariant properties that  can be tested using $\poly(1/\eps)$ queries.
For monotone properties that can be characterized by freeness from solutions to a family
$\calf$ of equations, we conjecture that there is a sharp dichotomy given by whether $\calf$ is infinite or
finite. It follows from Theorem \ref{thm:ocf1} and the discussion in the previous item that when $\calf$ is infinite,
the query complexity is $\poly(1/\eps)$. When $\calf$ is finite and the property is nontrivial, then the property is
equivalent to being free of solutions to a single equation $x_1 + \cdots + x_k = 0$ for
some odd integer $k > 1$. In this case, we conjecture that the query complexity is
super-polynomial, although the current best lower bound is only slightly non-trivial:
$\Omega(1/\eps^{2.423})$ for testing triangle-freeness \cite{BX10}. For non-monotone
properties characterized by freeness from solutions to a family of equations, \cite{CSX11}
showed that $(C_3, 110)$-freeness can be testing using $O(1/\eps^2)$ queries (recalling
the notation in Section \ref{sec:intro}), but there is no systematic understanding at
present of when $\poly(1/\eps)$ query complexity is possible for larger equations or for
arbitrary intersections of such non-monotone properties.  For properties
characterized by freeness from solutions to a system of equations of rank greater than
one, even less is known.

\item Another open problems left open by our results is whether the ${\tilde
 O}(1/\eps^2)$ bound for odd-cycle-freeness is tight. This is indeed the case for
bipartiteness testing in graphs \cite{BT04}, but a direct analogue of their hard instances
does not seem to work in our case.

\item One could also ask Question \ref{ques:canpoly} for linear-invariant properties that are
not subspace-hereditary. Given a linear-invariant property $\calp$, we say that a tester
$T$ is canonical for $\calp$ if there exists a fixed linear-invariant property $\calp'$
(not necessarily the same as $\calp$) such that when
$T$ is given oracle access to a function $f: \F_2^n \to \zo$, it operates by choosing
uniformly at random a subspace $H \leq \F_2^n$ and accepting if and only if $f$
restricted to $H$ satisfies the property $\calp'$.  Notice that unlike the
subspace-hereditary case, the canonical tester now need not be one-sided.  The stronger
form of Question \ref{ques:canpoly} is whether it is the case that for every
linear-invariant property $\calp$, there exists a canonical tester for $\calp$ with query
complexity $\poly(q(n,\eps))$ whenever $\calp$ is testable with query complexity
$q(n,\eps)$ by some tester. Goldreich and Trevisan \cite{GoldreichTrevisan} showed the
existence of such a canonical tester with polynomial blowup for graph properties.

\end{itemize}

\bibliographystyle{alpha}
\bibliography{testing}
\end{document}